\documentclass[english,12pt]{article}
\usepackage[T1]{fontenc}
\usepackage{geometry}
\usepackage{array,arydshln}
\usepackage{comment}
\usepackage{xcolor,soul}

\usepackage{graphicx, graphics}
\usepackage{enumitem}
\usepackage{amsmath,amsthm,amssymb,setspace,bm}

\usepackage[T1]{fontenc}
\usepackage[utf8]{inputenc}

\usepackage[colorlinks,citecolor=blue,urlcolor=blue,allcolors=blue,bookmarks=false,hypertexnames=true]{hyperref} 
\usepackage[square,authoryear]{natbib}
\usepackage{tikz}

\def\PC{\mbox{P\&C} }

\newtheorem{prop}{Proposition}

\newtheorem{thm}{Theorem}
\theoremstyle{definition}

\newtheorem{definition}{Definition}

\newtheorem{claim}{Claim}

\title{{Sequential Pricing Mechanisms for Surplus Division}\footnote{We are grateful to Pierre Fleckinger, Jeanne Hagenbach, Frédéric Koessler, Raphaël Lévy, and Nicolas Schutz, as well as to seminar audiences at Collegio Carlo Alberto, CREST, HEC Paris, Osaka University, Tokyo University of Science, Waseda University, and the 8th Spain-Japan Meeting on Economic Theory, for helpful comments. 
This work has benefited from a State grant managed by the Agence Nationale de la Recherche under the Investissements d'Avenir programme with the reference ANR-18-EURE-0005 / EUR DATA EFM and ANR-11-IDEX-0003/Labex Ecodec/ANR-11-LABX-0047.}}
\author{
	Yukihiko Funaki\thanks{
	{School of Political Science and Economics, Waseda University, Tokyo, Japan. E-mail: \texttt{funaki@waseda.jp}}.
    ORCID: 0000-0001-5202-7182.}
		\and
	Yukio Koriyama\thanks{ \raggedright
	Centre de Recherche en Économie et Statistiques (CREST), Institut Polytechnique de Paris, Palaiseau, France. 
	E-mail: \texttt{yukio.koriyama@polytechnique.edu}.
    ORCID: 0000-0002-3889-8317.}
		\and
	Matías Núñez\thanks{
	{Centre de Recherche en Économie et  Statistiques (CREST), CNRS, Institut Polytechnique de Paris, Palaiseau, France. E-mail: \texttt{matias.nunez@ensae.fr}}.
    ORCID: 0000-0003-0948-7647.}
		\and
	Giacomo Rostagno\thanks{{HEC Paris, Jouy-en-Josas, France, and RBB Economics, London, UK. E-mail: \texttt{giacomo.rostagno@rbbecon.com}}. The views expressed are those of the author and do not reflect those of RBB Economics or its clients.}
}  

\begin{document}

\maketitle

\begin{abstract}
Extending the Price-and-Choose (P\&C) mechanism of \citet{Echenique2025}, we propose the Price-Accept-and-Choose (PA\&C) mechanism, which preserves efficiency while eliminating P\&C's first-mover advantage.
We then analyze randomized and bidding variants and show that the resulting equilibrium payoffs correspond to standard solutions in transferable utility games: the Center of the Imputation Set value for P\&C and the Shapley value for PA\&C.
In the randomized variants, these solutions arise in expectation; in the bidding variants, they are implemented on every equilibrium path.
We further relate other efficient designs, such as balanced VCG payments, to additional solution concepts, forging an interpretable bridge between implementation theory and cooperative game theory.

\end{abstract}

\section{Introduction}

When a group chooses among several alternatives, selecting the efficient one is arguably the best choice when monetary transfers are available: picking an efficient alternative ensures that total welfare is maximized while the transfers make it feasible to compensate those who lose from the choice. This is the logic behind the Welfare theorems, the Vickrey-Clarke-Groves mechanism, and much of mechanism design and implementation theory.\footnote{\cite{Moore1988} roughly state that for \emph{any} social choice
rule, one can design a mechanism that yields unique implementation in subgame-perfect equilibria with transfers.} But efficiency alone does not pin down how the surplus is to be divided as two mechanisms achieving efficiency may distribute the gains very differently. We study how specific mechanisms shape surplus division in social choice problems with transfers and complete information. Our main finding is that simple sequential pricing mechanisms, in which players propose prices over alternatives and others choose or reject, generate equilibrium payoffs that correspond to {canonical solution concepts in transferable utility (TU) games}. The Price-and-Choose mechanism of \cite{Echenique2025} yields the Center of the Imputation Set value; our proposed extension, in which players may reject proposed prices, yields the Shapley value. These correspondences reflect how each mechanism allocates bargaining power across positions in the sequence of play.

To fix ideas, consider two sisters, Sophia and Victoria, who have inherited a firm and must decide how to restructure it. Several alternatives are available, such as selling, merging with a partner, or keeping and reorganizing, and each sister values these differently. They know each other's preferences, but an arbitrator who must approve the decision does not. Under P\&C, Sophia proposes a price for each alternative (the prices adding up to zero), and Victoria selects whichever alternative she prefers and pays Sophia the associated price. \cite{Echenique2025} show that every subgame-perfect equilibrium is efficient: Sophia sets prices that make Victoria indifferent, so Victoria is willing to pick whichever alternative Sophia prefers, which coincides with the surplus-maximizing one. But the surplus division is quite unbalanced as Victoria receives only her average utility across alternatives, while Sophia captures the rest. 

Now suppose Victoria can reject the proposed prices and simply choose her favorite alternative with no transfer. This is PA\&C. The rejection option raises Victoria's guaranteed payoff from her average utility to her maximum utility,  altering the surplus division. Anticipating that Victoria can reject the prices, Sophia now offers prices generous enough that Victoria prefers to accept and choose the efficient alternative rather than reject and choose her own favorite. In equilibrium, each sister receives exactly her marginal contribution to the joint surplus.

The logic of the Sophia and Victoria's example extends to any number of players and alternatives. In the general PA\&C mechanism, players move sequentially: each proposes a price vector to the next player, who may accept or reject all previously proposed prices before submitting her own proposal. The last player accepts or rejects and then selects an alternative. If any player rejects, all previously proposed prices are erased and the continuation game is itself a PA\&C mechanism among the remaining players. This means that each player can guarantee herself at least what she would earn as the first mover in the smaller game — her marginal contribution to the coalition formed by herself and all subsequent players. We show that in every subgame-perfect equilibrium, each player receives exactly her marginal contribution to the coalition of players who follow her, and the selected alternative is efficient. This full-implementation result, according to which every equilibrium is efficient, and every efficient alternative is supported by some equilibrium, holds for any finite number of players and alternatives.

While these mechanisms ensure efficiency, the division of surplus depends on the sequence of play.  To address this, we first consider randomized versions in which the order of moves is chosen uniformly at random. We show that under randomization the expected payoffs correspond to well-known solution concepts of the associated cooperative game: the Center of the Imputation Set (CIS, \citet{DriessenFunaki1991}) for the P\&C, and the Shapley value (\citet{Shapley1953}) for the PA\&C.  The CIS assigns each player her standalone payoff plus an equal share of the remaining surplus, while the Shapley value averages each player’s marginal contribution over all possible positions in the sequence.

Randomization yields a clear ex‑ante interpretation of surplus division, but ex‑post payoffs still depend on the realized order. To endogenize the order of play, we analyze bidding procedures. For P\&C, we rely on the bidding game of \citet{Echenique2025} and show that its equilibrium payoffs coincide with the CIS. For PA\&C, we design a new bidding procedure that compensates for positional disadvantages and yields the Shapley value in equilibrium. Unlike in P\&C, where only the first position determines the surplus division, in PA\&C each position affects payoffs because players can reject earlier price vectors. Our procedure accounts for this feature and fully implements the Shapley value in a setting where order matters for all players.

Beyond these core mechanisms, we study two additional designs. First, we consider a variant of PA\&C in which the last mover cannot accept or reject the preceding price vector. Following the logic of P\&C, the last player must then be indifferent across all alternatives, and the resulting payoffs correspond to the Shapley value of the cooperative game associated with the CIS. Second, we examine balanced Vickrey-Clarke-Groves (VCG) payments (\citet{Vickrey1961,Clarke1971,Groves1973}) and show that they implement the Egalitarian Non‑Separable Contribution (ENSC) solution. The ENSC assigns each player her separable contribution to welfare and divides the remaining surplus equally. 
Moreover, ENSC is the dual of CIS.
These links highlight how efficient mechanisms can be tailored to achieve specific and interpretable surplus divisions.

The rest of the paper is organized as follows: a literature review is provided in the rest of the section. Section \ref{sec:Setting} introduces the theoretical framework, while Section \ref{sec:PAC} analyses the PA\&C. Section \ref{sec:Cooperative Games} studies the relationship between the P\&C, PA\&C SPE payoffs and the solution concepts in corresponding TU games. Section \ref{sec:Further} studies the extension of PA\&C and the balanced VCG payments. Finally, Section \ref{sec:Conclusion} concludes.

\subsection{Literature Review} \label{sec:Literature}

Our paper contributes to several strands of the literature. Most directly, it relates to the theory of implementation. Without transfers, efficient implementation is severely constrained: \cite{Hurwicz_1978} show that dictatorship is the only Pareto-efficient implementable rule in two-person societies, and more generally, \cite{Maskin99} establishes that monotonicity is necessary for Nash implementation, ruling out many desirable rules. When transfers are allowed, these constraints largely disappear. \cite{Moore1988} show that under quasi-linear preferences, any social choice rule can be subgame-perfect implemented. \cite{Aghion2012} have criticized this general result by showing that finite and countable mechanisms fail to achieve full implementation as soon as we depart from the complete information. \cite{Echenique2025} propose the Price-and-Choose mechanism (with infinite strategy spaces), which achieves efficient implementation through a sequential structure: one player sets prices, another chooses. We build on P\&C by introducing an accept-or-reject stage that preserves its simplicity and efficiency while altering the division of surplus.

The original \PC mechanism is reminiscent of the ``Divide and Choose'' paradigm used in allocation problems (e.g., \citet{Nicolo_Divide}), where players divide a surplus and choose among divisions. While our setting differs in that we deal with general social choice environments rather than divisible goods, the underlying concern, how to split the gains from efficient outcomes, remains central. To address this, we study how different mechanisms induce different surplus allocations and interpret the resulting equilibrium payoffs through the lens of solution concepts in cooperative game theory.

In particular, we connect the equilibrium outcomes of the mechanisms to canonical transferable utility (TU) game solution concepts. For instance, we show that the randomized version of PA\&C yields expected payoffs corresponding to the Shapley value (\citet{Shapley1953}), while the randomized version of P\&C corresponds to the Center of the Imputation Set (\citet{DriessenFunaki1991}). These insights highlight how different mechanisms implement not only efficient outcomes, but also different cooperative notions of fairness.

Our paper also builds on work that uses bidding procedures to implement TU-game solutions. \citet{PEREZCASTRILLO2001} introduce a non-cooperative mechanism to implement the Shapley value via sequential bidding. \citet{BrinkFunaki2015SCW} extend this approach in the context of discounted values. While these papers focus solely on surplus division, we embed bidding within a broader social choice framework where agents choose among alternatives. We show that by combining a bidding stage with an efficient dynamic mechanism, one can implement both efficiency and a cooperative solution concept either in expectation or ex post.

A closely related literature studies multi-bidding mechanisms in precisely our social choice setting. \citet{PerezCastrilloWettstein2002} propose a mechanism in which agents simultaneously submit a bid for every project together with an announcement of their preferred one, and show that every Nash equilibrium selects an efficient project. \citet{Ehlers2009} removes the project announcement and shows that existence of equilibrium is then no longer guaranteed, while \citet{Kamijo2014} extends the design to environments in which the cost of the selected project must be covered by the agents themselves. These mechanisms share our objective of implementing an efficient alternative in a social choice problem with transfers, but they differ from our sequential deterministic mechanisms in that they are static and they rely on lotteries.

Our focus on simple and intuitive mechanisms also connects with a growing body of experimental work that investigates the practical performance of implementation designs. Specifically, \citet{Funaki_Exp} provide an empirical assessment of both the P\&C and the PA\&C mechanisms. In a setting with two players and two alternatives, the paper shows that both P\&C and PA\&C achieve a high efficiency rate and demonstrate consistent subject behavior across protocols, even when compared to the well-known Ultimatum Game. This experimental validation underscores the practical viability of these simple dynamic mechanisms.

\section{The Model} \label{sec:Setting}

\paragraph{Primitives.} 
Let $N = \{1,2,\dots,n\}$, with $n \ge 2$, denote a finite set of individuals bargaining over a finite set of alternatives 
$\mathcal{A} = \{a_1, a_2, \dots, a_m\}$, where $m :=|\mathcal{A}| \ge 2$. 
Each individual $i$'s preferences over alternatives are represented by a utility function 
$u_i : \mathcal{A} \to \mathbb{R}$. 
Without loss of generality, we assume that $u_i(a) \ge 0$ for all $i \in N$ and all $a \in \mathcal{A}$.\footnote{This assumption guarantees that the associated TU-game is monotonic, so that each player makes a non-negative contribution to every coalition. Although not required for our implementation results, it simplifies the TU-game interpretation.} 
Let $u = (u_i(a))_{i \in N,\, a \in \mathcal{A}} \in \mathbb{R}^{n \times m}$ denote the utility profile. 
A social choice problem is therefore represented by the pair $(\mathcal{A}, u)$.

Monetary transfers $t = (t_1, \dots, t_n) \in \mathbb{R}^n$ are available. If alternative $a \in \mathcal{A}$ is selected and individual $i$ receives transfer $t_i$, her total utility takes the quasi-linear form $u_i(a) + t_i$.

\paragraph{Efficiency.}
For each $i \in N$, let $\mathcal{A}_i^{\ast}$ denote the set of efficient alternatives for the coalition of players $\{i, i+1, \ldots, n\}$, defined as:
\begin{equation}
    \mathcal{A}_i^{\ast}
    = \arg\max_{a \in \mathcal{A}} \sum_{j \ge i} u_j(a).
    \label{eq:a_i_star}
\end{equation}
In particular, the set of socially efficient alternatives is given by $\mathcal{A}_1^{\ast}$, which we denote simply by $\mathcal{A}^{\ast}$.
Note that the collection of sets $(\mathcal{A}_i^{\ast})_{i \in N}$ depends on the ordering of individuals.

\paragraph{Marginal Contribution.}
For each $i \in N$, let $x_i^{\ast}$ denote the \emph{marginal contribution} of individual $i$, defined by:
\begin{equation} \label{eq:x_star}
x^{\ast}_{i} =
\begin{cases} 
\displaystyle
\sum_{j \ge i} u_j(a_{i}^{\ast}) - \sum_{j > i} u_j(a_{i+1}^{\ast})
& \text{for all } i \in N \setminus \{n\}, \text{ with } a_i^{\ast} \in \mathcal{A}_i^{\ast} \text{ and } a_{i+1}^{\ast} \in \mathcal{A}_{i+1}^{\ast}, \\[1em]
u_n(a_n^{\ast})
& \text{for } i = n, \text{ with } a_n^{\ast} \in \mathcal{A}_n^{\ast}.
\end{cases}
\end{equation}
For each $i \in N$, $x_i^{\ast}$ is well-defined, as it does not depend on the particular choices of $a_i^{\ast}$, $a_{i+1}^{\ast}$, or $a_n^{\ast}$.

Intuitively, $x_i^{\ast}$ measures the marginal surplus generated by adding player $i$ to the coalition $\{i+1, \ldots, n\}$. The term $\sum_{j \ge i} u_j(a_i^{\ast})$ represents the total utility of the coalition when player $i$ joins $\{i+1, \ldots, n\}$ and an efficient alternative $a_i^{\ast}$ is selected, whereas $\sum_{j > i} u_j(a_{i+1}^{\ast})$ is the total utility of the coalition $\{i+1, \ldots, n\}$ under an efficient alternative $a_{i+1}^{\ast}$. Their difference therefore captures the surplus that player $i$ can bring. For the last player $n$, this expression reduces to her maximum attainable utility, $u_n(a_n^{\ast})$.

\paragraph{Subgame-Perfect Implementation}
A mechanism induces an extensive-form game. We say that a mechanism \emph{subgame-perfect implements} the set of efficient alternatives $\mathcal{A}^{\ast}$ if, for any profile of utility functions $u$, the following two conditions are satisfied in the game induced by the mechanism:
\begin{itemize}
    \item[(i)] For every subgame-perfect equilibrium, the selected alternative belongs to $\mathcal{A}^{\ast}$.
    \item[(ii)] For every alternative in $\mathcal{A}^{\ast}$, there is a subgame-perfect equilibrium that selects it.
\end{itemize}

\section{Price, Accept and Choose} \label{sec:PAC}

We propose the Price-Accept-and-Choose (PA\&C) mechanism, which builds on the
Price-and-Choose (P\&C) mechanism of \citet{Echenique2025}. Like P\&C, PA\&C
achieves subgame-perfect implementation of the efficient alternative. Section
\ref{sec:Cooperative Games} compares the equilibrium payoffs of the two mechanisms
using cooperative-game solution concepts.

\subsection{Two-Player Case}

The PA\&C protocol with two players proceeds as follows. Player 1 (the proposer) proposes a price vector $p \in \mathbb{R}^m$, assigning a price to each alternative, subject to the balanced-budget constraint that prices sum to zero across alternatives. Player 2 (the chooser) then either \textit{accepts} or \textit{rejects} the proposed price vector. If accepted, player 2 chooses an alternative $a_k \in \mathcal{A}$ and makes the transfer $p(a_k)$ to player 1. If rejected, player 2 chooses an alternative from $\mathcal{A}$ and no transfer is made.

Let $P$ denote the set of balanced-budget price vectors:
\begin{equation*}
    P = \left\{ p \in \mathbb{R}^m \;\middle|\; \sum_{k=1}^m p(a_k) = 0 \right\}.
\end{equation*}

Suppose the proposed price vector is $p \in P$, and let player 2's action be $c_2 \in C = \{A, R\}$, where $A$ denotes acceptance and $R$ denotes rejection.
The realized price vector $p^c \in P$ is then given by:
\begin{equation*}
    p^c = \begin{cases}
        p      & \text{if } c_2 = A, \\
        \mathbf{0} & \text{if } c_2 = R.
\end{cases}
\end{equation*}
The timing of the PA\&C mechanism is as follows:
\begin{enumerate}
    \item Player 1 proposes a price vector $p \in P$.
    \item Player 2 either accepts ($A$) or rejects ($R$) the price vector.
    \item Player 2 chooses an alternative $a \in \mathcal{A}$ and transfers $p^c(a)$ to player 1.
\end{enumerate}

A strategy profile of the game induced by the PA\&C mechanism is a pair
$\sigma = (\sigma_1, \sigma_2)$, where player 1's strategy is a price vector
$\sigma_1 \in P$ and player 2's strategy is a function $\sigma_2 : P \to C \times \mathcal{A}$.
Given a price vector $p$ proposed by player 1 and the strategy subsequently selected by
player 2, the players' payoffs are:
\begin{equation*}
    y = \left(y_1, y_2\right) = \left(u_1(a) + p^c(a), u_2(a) - p^c(a)\right).
\end{equation*}

\begin{prop} \label{prop:SPE_PAC_2}
Consider the PA\&C mechanism with two players. Then:
\begin{enumerate}[label=\roman*)]
    \item In every subgame perfect equilibrium, the equilibrium payoff of each player $i \in N$ coincides
    with the marginal contribution $x_i^{\ast}$.
    \item The PA\&C mechanism subgame-perfect implements the set of efficient alternatives.
\end{enumerate}
\end{prop}

\begin{proof}[Proof of Proposition \ref{prop:SPE_PAC_2}]

Let $y = (y_1, y_2) \in \mathbb{R}^2$ be any SPE payoff vector. We establish three claims.
Claim~\ref{cl:yi_le_xi} shows that each player's payoff is at least her marginal contribution.
Claim~\ref{cl:yi_eq_xi} strengthens this to equality and shows that player 2 selects an efficient alternative, establishing part~(i).
Claim~\ref{cl:a_star} shows that every efficient alternative is supported by some SPE, establishing part~(ii).

\begin{claim} \label{cl:yi_le_xi}
$y_i \geq x_i^{\ast}$ for each $i \in N$.
\end{claim}
\begin{proof}[Proof of Claim \ref{cl:yi_le_xi}]
Player 2 can guarantee $x_2^{\ast}$ by rejecting and choosing any $a_2^{\ast} \in \mathcal{A}_2^{\ast}$, so $y_2 \geq x_2^{\ast}$.

For player 1, suppose for contradiction that $y_1 < x_1^{\ast}$, so that $x_1^{\ast} - y_1 = \varepsilon > 0$. Fix any $a^{\ast} \in \mathcal{A}^{\ast}$ and consider the price vector $p \in P$ defined by
\[
p(a^{\ast}) = u_2(a^{\ast}) - x_2^{\ast} - \tfrac{\varepsilon}{2}, \qquad
p(a) = -\frac{1}{m-1}\!\left(u_2(a^{\ast}) - x_2^{\ast} - \tfrac{\varepsilon}{2}\right) \text{ for } a \neq a^{\ast}.
\]
Under this price vector, choosing $a \neq a^{\ast}$ yields player 2 strictly less than $x_2^{\ast}$, while choosing $a^{\ast}$ yields strictly more than $x_2^{\ast}$; hence player 2 accepts and selects $a^{\ast}$. Player 1's payoff is then $x_1^{\ast} - \frac{\varepsilon}{2} > y_1$, a contradiction.
\end{proof}

\begin{claim} \label{cl:yi_eq_xi}
$y_i = x_i^{\ast}$ for each $i \in N$, and player 2 selects an alternative in $\mathcal{A}^{\ast}$.
\end{claim}
\begin{proof}[Proof of Claim \ref{cl:yi_eq_xi}]
Let $a'$ be the alternative chosen in equilibrium. Since transfers cancel,
\[
y_1 + y_2 = u_1(a') + u_2(a').
\]
By Claim~\ref{cl:yi_le_xi}, $y_1 + y_2 \geq x_1^{\ast} + x_2^{\ast} = u_1(a^{\ast}) + u_2(a^{\ast})$ for any $a^{\ast} \in \mathcal{A}^{\ast}$. By efficiency of $\mathcal{A}^{\ast}$, $u_1(a') + u_2(a') \leq u_1(a^{\ast}) + u_2(a^{\ast})$. These two inequalities force equality throughout, giving $y_1 + y_2 = x_1^{\ast} + x_2^{\ast}$. Combined with $y_i \geq x_i^{\ast}$ for each $i$, we obtain $y_i = x_i^{\ast}$ for each $i \in N$. Moreover, $u_1(a') + u_2(a') = u_1(a^{\ast}) + u_2(a^{\ast})$ implies $a' \in \mathcal{A}^{\ast}$.
\end{proof}

\begin{claim} \label{cl:a_star}
For each $a^{\ast} \in \mathcal{A}^{\ast}$, there exists an SPE in which $a^{\ast}$ is chosen.
\end{claim}
\begin{proof}[Proof of Claim \ref{cl:a_star}]
Fix any $a^{\ast} \in \mathcal{A}^{\ast}$ and consider the price vector $p \in P$ defined by
\[
p(a^{\ast}) = u_2(a^{\ast}) - x_2^{\ast}, \qquad
p(a) = -\frac{1}{m-1}\!\left(u_2(a^{\ast}) - x_2^{\ast}\right) \text{ for } a \neq a^{\ast}.
\]
By definition of $x_2^{\ast}$, we have $p(a^{\ast}) \leq 0$ and $p(a) \geq 0$ for $a \neq a^{\ast}$, so $u_2(a) - p(a) \leq x_2^{\ast}$ for all $a$. Accepting and choosing $a^{\ast}$ yields player 2 exactly $x_2^{\ast}$, which weakly dominates all other actions (accept another alternative, or reject). Hence accepting and choosing $a^{\ast}$ is a best response for player 2, and player 1's payoff is $x_1^{\ast}$.

If player 1 deviates to induce a different efficient alternative $a' \in \mathcal{A}^{\ast}$, Claim~\ref{cl:yi_eq_xi} gives $y_1' = x_1^{\ast}$, so the deviation is not profitable. Hence $a^{\ast}$ is supported by an SPE.
\end{proof}

\noindent Claims \ref{cl:yi_le_xi}--\ref{cl:a_star} establish Proposition~\ref{prop:SPE_PAC_2}.
\end{proof}

In contrast to P\&C, the PA\&C mechanism eliminates the first-mover advantage, transferring the strategic advantage to the second mover. This shift is driven by player 2's option to reject the proposed price vector, thereby guaranteeing herself a payoff equal to her standalone value $x_2^{\ast}$. As we show in Section \ref{sec:Cooperative Games}, this rejection option systematically favors the last mover, a crucial distinction from P\&C.

\subsection{Many-Player Case}

We now extend the results of Proposition~\ref{prop:SPE_PAC_2} to an arbitrary number
of players, showing that its conclusions continue to hold in the $n$-player PA\&C mechanism.

The mechanism proceeds sequentially from player 1 to player $n$. Player 1 publicly proposes a balanced-budget price vector $p_1 \in P$ to player 2. 
For each $i = 2, \dots, n-1$, player $i$ observes all previous moves, decides whether to accept ($A$) or reject ($R$) all previously proposed price vectors, and then submits her own proposal $p_i \in P$ to player $i+1$.
Finally, player $n$ accepts or rejects the current price proposal and chooses an alternative $a \in \mathcal{A}$. Only accepted price proposals are implemented.

Let $h_i$ denote the history observed by player $i$ before she chooses her action $c_i \in \{A,R\}$, consisting of
the sequence:
\[
h_i = \left(p_1, c_2, p_2, c_3, p_3, \dots, c_{i-1}, p_{i-1}\right),
\]
and let $H_i$ be the set of all such histories. For each $i < n$, player $i$'s
pure behavioral strategy is a function:
\[
\sigma_i : H_i \to C \times P,
\]
where the first component specifies whether player $i$ accepts or rejects the current
proposal, and the second specifies the price vector she proposes to player $i+1$.
We write $\sigma_i(h_i) = \bigl(c_i(h_i),\, p_i(h_i)\bigr)$.

Player $n$'s strategy is a function
$
\sigma_n : H_n \to C \times \mathcal{A},
$
where she first accepts or rejects the proposed prices and then selects an alternative
$a \in \mathcal{A}$.

Consider an equilibrium vector of proposed prices $p = (p_i)_{i=1}^{n-1}$. Each
player $i > 1$ accepts or rejects all previously proposed price vectors. The
realized price vector need not coincide with the proposed one, since any player
may have rejected prior proposals.

The realized price vector for player $i$, denoted $p_i^c$, determines the actual
transfers once player $n$ selects an alternative $a \in \mathcal{A}$. It is defined by:
\begin{equation*}
    p^c_i =
    \begin{cases}
        p_i & \text{if } c_j = A \text{ for all } j > i, \\
        \mathbf{0} & \text{otherwise,}
    \end{cases}
\end{equation*}
for all $i \in \{1, \dots, n-1\}$. That is, $p_i$ is implemented only if every
subsequent player accepts.

Given alternative $a \in \mathcal{A}$, transfers are made sequentially: if all
prices are accepted, player $n$ pays $p_{n-1}^c(a)$ to player $n-1$, who pays
$p_{n-2}^c(a)$ to player $n-2$, and so on, with player 2 paying $p_1^c(a)$ to
player 1. With quasi-linear preferences, the resulting payoffs are:
\begin{align*}
    y_n(p^c, a) &= u_n(a) - p^c_{n-1}(a), \\
    y_i(p^c, a) &= u_i(a) - p^c_{i-1}(a) + p^c_i(a) \quad \text{for } i = 2, \dots, n-1, \\
    y_1(p^c, a) &= u_1(a) + p^c_1(a).
\end{align*}

\begin{thm} \label{prop:PAC_General}
Consider the PA\&C mechanism with $n \geq 2$ players. Then:
\begin{enumerate}[label=\roman*)]
    \item In every subgame perfect equilibrium, the equilibrium payoff of each player $i \in N$ coincides
    with her marginal contribution $x_i^{\ast}$.
    \item The PA\&C mechanism subgame-perfect implements the set of efficient alternatives.
\end{enumerate}
\end{thm}

\begin{proof}
The proof proceeds by induction on $n$. The base case $n = 2$ is Proposition~\ref{prop:SPE_PAC_2}. Assume the proposition holds for $n - 1$ players.

Fix an arbitrary SPE of the $n$-player PA\&C mechanism. Let $y = (y_1, \dots, y_n) \in \mathbb{R}^n$ be the equilibrium payoff vector and $a$ the equilibrium outcome. Since transfers cancel,
\begin{equation} \label{eq:sum_payoffs}
    \sum_{i \in N} y_i = \sum_{i \in N} u_i(a).
\end{equation}

\textbf{Claim 1:} $y_i \geq x_i^{\ast}$ for each $i \geq 2$.

If player $i \geq 2$ rejects, the continuation game is a PA\&C mechanism with $n - (i-1)$ players. By the induction hypothesis, player $i$ obtains $x_i^{\ast}$ in that subgame, so she can always guarantee at least $x_i^{\ast}$.

\textbf{Claim 2:} $y_1 \geq x_1^{\ast}$.

Suppose for contradiction that $y_1 < x_1^{\ast}$, and let $\varepsilon := x_1^{\ast} - y_1 > 0$. Fix any $a^{\ast} \in \mathcal{A}^{\ast}$ and define:
\[
u_E := \sum_{j \geq 2} u_j(a^{\ast}) - \sum_{j \geq 2} x_j^{\ast}.
\]
Since $\sum_{j \geq 2} x_j^{\ast} = \max_{a \in \mathcal{A}} \sum_{j \geq 2} u_j(a) \geq \sum_{j \geq 2} u_j(a^{\ast})$, we have $u_E \leq 0$. Consider the deviation $\tilde{p}_1 \in P$ defined by
\[
\tilde{p}_1(a^{\ast}) = u_E - \frac{\varepsilon}{2}, \qquad
\tilde{p}_1(a) = -\frac{1}{m-1}\!\left(u_E - \frac{\varepsilon}{2}\right) \text{ for } a \neq a^{\ast}.
\]
Since $u_E \leq 0$, we have $\tilde{p}_1(a) > 0$ for all $a \neq a^{\ast}$, making $a^{\ast}$ the unique maximizer of $\sum_{j \geq 2} u_j(a) - \tilde{p}_1(a)$. By the induction hypothesis applied to the continuation game after player 2 accepts $\tilde{p}_1$, player 2 accepts (since doing so yields strictly more than $x_2^{\ast}$) and $a^{\ast}$ is selected. A straightforward calculation then gives that player 1's payoff from $\tilde{p}_1$ is $x_1^{\ast} - \frac{\varepsilon}{2} > y_1$, contradicting equilibrium. Hence $y_1 \geq x_1^{\ast}$.

\textbf{Claim 3:} $y_i = x_i^{\ast}$ for all $i \in N$.

By~\eqref{eq:sum_payoffs} and efficiency of $\mathcal{A}^{\ast}$,
\[
\sum_{i \in N} y_i = \sum_{i \in N} u_i(a) \leq \sum_{i \in N} u_i(a^{\ast}) = \sum_{i \in N} x_i^{\ast}.
\]
Combined with $y_i \geq x_i^{\ast}$ for all $i$ (Claims~1 and~2), this forces $y_i = x_i^{\ast}$ for every $i \in N$.

\textbf{Claim 4:} Every SPE selects some $a^{\ast} \in \mathcal{A}^{\ast}$, and every $a^{\ast} \in \mathcal{A}^{\ast}$ is selected by some SPE.

The first part follows immediately: if the equilibrium outcome $a \notin \mathcal{A}^{\ast}$, then $\sum_{i \in N} u_i(a) < \sum_{i \in N} x_i^{\ast}$, contradicting Claim~3 via~\eqref{eq:sum_payoffs}.

For the second part, fix any $a^{\ast} \in \mathcal{A}^{\ast}$ and consider strategies under which ties among efficient alternatives are broken in favor of $a^{\ast}$. Since all alternatives in $\mathcal{A}^{\ast}$ generate the same total surplus, this tie-breaking leaves every player's marginal contribution unchanged. By Claim~3, no player has a profitable deviation, so the resulting strategy profile is an SPE that selects $a^{\ast}$.
\end{proof}

\section{Relationship with TU-game Solutions} \label{sec:Cooperative Games}

Both the PA\&C and P\&C mechanisms subgame-perfect implement the set of efficient alternatives. While efficiency ensures that total surplus is maximized, the distribution of this surplus depends on the mechanism employed. 
To analyze equilibrium payoffs, we draw on classical solution concepts for the transferable-utility (TU) games. 
We show that the PA\&C and P\&C mechanisms correspond to two canonical solutions of the associated TU game: the \emph{Shapley value} \citep{Shapley1953} and the \emph{Center of the Imputation Set value} \citep{DriessenFunaki1991}, respectively.

We associate two TU games with each social choice problem $(\mathcal{A}, u)$.
{In both, a coalition $S$ with $|S| \ge 2$ that takes the collective decision selects the alternative maximizing its aggregate utility and redistributes the resulting surplus among its members through transfers, so that its worth is $\max_{a \in \mathcal{A}} \sum_{i \in S} u_i(a)$.\footnote{Equivalently, each coalition is assumed to be able to secede and take the collective decision into its own hands, its members being free to redistribute the resulting aggregate utility among themselves. In our mechanisms this is not merely a modeling convention: when player $i$ rejects the proposed prices, the preceding players are removed from the transfer system and the remaining coalition $\{i, i+1, \dots, n\}$ does precisely this.} The two games differ only in how they treat the standalone worth.}

\begin{definition}[Game $v$] \label{def:v}
Given a social choice problem $(\mathcal{A}, u)$, the TU game $v$ is defined by:
\begin{equation} \label{eq:def_v}
    v(S) = \max_{a \in \mathcal{A}} \sum_{i \in S} u_i(a),
    \qquad \text{for all } S \subseteq N.
\end{equation}
\end{definition}

\begin{definition}[Game $\tilde{v}$]
\label{def:v_tilde}
Given a social choice problem $(\mathcal{A}, u)$, the TU game $\tilde{v}$ is defined by:
\begin{equation*} 
\tilde{v}(S) =
\begin{cases}
v(S), & \text{if } |S| \geq 2, \\[4pt]
\displaystyle \frac{1}{m} \sum_{k=1}^m u_i(a_k) =: \mathrm{Avg}_i, & \text{if } S = \{i\}.
\end{cases}
\end{equation*}
\end{definition}

{The two specifications of the singleton worth correspond to two standard benchmarks for the position of a player who is party to no agreement. 
Under $v$, the singleton worth is given by the same expression as that of any other coalition, $v(\{i\}) = \max_{a \in \mathcal{A}} u_i(a)$: player $i$ can secure her most preferred alternative on her own, as if she held a veto over the collective decision. 
Under $\tilde{v}$, player $i$ has no influence whatsoever over which alternative is selected and faces all of them symmetrically, so that her worth is her expected utility when the alternative is chosen uniformly at random, $\tilde{v}(\{i\}) = \mathrm{Avg}_i$. 
The former is the standalone worth of a player endowed with full individual bargaining power; the latter, that of a player endowed with none.
Our results show that the choice between these two benchmarks, and hence between the two solution concepts they generate, is exactly what is at stake in the design choice of whether players may reject the proposed prices.}

\subsection{Price and Choose}

We formally define the P\&C mechanism of \citet{Echenique2025}. The mechanism
follows the same sequential structure as PA\&C, with one crucial difference:
players have no option to accept or reject the previously proposed price vector.
The timing is as follows:
\begin{enumerate}[label=\roman*)]
    \item Player 1 proposes a price vector $p_1 \in P$.
    \item For each $i = 2, \dots, n-1$, player $i$ proposes a price vector $p_i \in P$
    to player $i+1$.
    \item Player $n$ chooses an alternative $a \in \mathcal{A}$.
    \item For each $i < n$, player $i+1$ transfers $p_i(a)$ to player $i$.
\end{enumerate}

\begin{prop}[\citet{Echenique2025}] \label{prop:PC_General}
Consider the P\&C mechanism with $n$ players. Then:
\begin{enumerate}[label=\roman*)]
    \item In every SPE, each player $i \geq 2$ obtains the payoff equal to $\mathrm{Avg}_i$, and player~1 obtains the payoff equal to $v(N) - \sum_{j \neq 1} \mathrm{Avg}_j$.
    \item The P\&C mechanism subgame-perfect implements the set of efficient alternatives.
\end{enumerate}
\end{prop}

In contrast to PA\&C, the P\&C mechanism exhibits a clear asymmetry among players.
Since no player can reject the proposed price vector, every player $i \geq 2$ is
rendered indifferent among all alternatives in equilibrium and obtains her average
payoff $\mathrm{Avg}_i$. Player~1 obtains the surplus $v(N) - \sum_{j \in N}
\mathrm{Avg}_j$ in addition to $\mathrm{Avg}_1$ in equilibrium, yielding her a
strict first-mover advantage. This advantage is formalized by the following result.

\begin{prop} \label{prop:first_mover}
The first mover's payoff in subgame perfect equilibrium is weakly higher under
P\&C than under PA\&C.
\end{prop}

\begin{proof}
By Theorem~\ref{prop:PAC_General}, the first mover's SPE payoff under PA\&C is
\[
x_1^{\ast} = \sum_{j \in N} u_j(a_1^{\ast}) - \sum_{j \geq 2} u_j(a_2^{\ast}).
\]
By Proposition~\ref{prop:PC_General}, her SPE payoff under P\&C is $v(N) -
\sum_{j \geq 2} \mathrm{Avg}_j$. By definition of $a_1^{\ast}$ and $a_2^{\ast}$,
we have $\sum_{j \in N} u_j(a_1^{\ast}) = v(N)$ and $\sum_{j \geq 2} u_j(a_2^{\ast})
\geq \sum_{j \geq 2} \mathrm{Avg}_j$. The result follows.
\end{proof}

\subsubsection{Randomized Price and Choose}

One natural way to correct for this asymmetry is the \emph{randomized} P\&C
mechanism, in which the order of play is drawn from a uniform distribution so
that each player has an equal probability of being the first mover.

\begin{definition}[Randomized P\&C mechanism]
A permutation of $\{1, 2, \ldots, n\}$ is drawn uniformly at random by Nature.
The P\&C mechanism is then played according to the realized order.
\end{definition}

The expected payoffs of the randomized P\&C mechanism coincide with the Center
of the Imputation Set (CIS) value \citep{DriessenFunaki1991} of the TU game
$\tilde{v}$.

\begin{definition}[CIS value]
The Center of the Imputation Set value is defined by
\begin{equation}
    \mathrm{CIS}_i(v) = v(\{i\}) + \frac{1}{n}\left(v(N) - \sum_{j \in N} v(\{j\})\right),
    \qquad \forall i \in N.
\end{equation}
\end{definition}

By randomizing the order of play, each player is equally likely to be the first
mover, and in all other positions she is guaranteed her average payoff. This is
precisely the intuition behind the CIS value, which awards each player $i$ her
standalone worth $\mathrm{Avg}_i$ plus an equal share of the surplus
$v(N) - \sum_{j \in N} \mathrm{Avg}_j$.

\begin{thm}
The expected payoff of any SPE of the randomized P\&C mechanism coincides with the CIS value of the TU game $\tilde{v}$.
\end{thm}

\begin{proof}
The expected payoff of player $i$ is the average of her payoffs across all
possible orders. Fix an arbitrary order. By Proposition~\ref{prop:PC_General},
when player $i$ is the first mover her SPE payoff is
\begin{equation*}
    \sum_{j \in N} u_j(a^{\ast}) - \sum_{j \neq i} \mathrm{Avg}_j
    = \tilde{v}(N) - \sum_{j \neq i} \tilde{v}(\{j\}),
\end{equation*}
and when she is not the first mover her payoff is $\mathrm{Avg}_i = \tilde{v}(\{i\})$.
Since each player is the first mover with probability $1/n$, player $i$'s expected
payoff is
$$
    \frac{1}{n}\!\left(\tilde{v}(N) - \sum_{j \neq i} \tilde{v}(\{j\})\right)
    + \left(1 - \frac{1}{n}\right)\tilde{v}(\{i\})
    = \tilde{v}(\{i\}) + \frac{1}{n}\!\left(\tilde{v}(N) - \sum_{j \in N} \tilde{v}(\{j\})\right),
$$
which coincides with $\mathrm{CIS}_i(\tilde{v})$.
\end{proof}

\subsubsection{Bid, Price and Choose}

While randomizing the order of play ensures ex-ante equal division of the surplus,
the ex-post allocation remains unequal once the order is realized. To achieve ex-post
equality, \citet{Echenique2025}, building on \citet{PEREZCASTRILLO2001}, introduce a
preliminary bidding stage that determines the order of play.

All players simultaneously submit bids for the right to be the first mover. The
highest bidder obtains the first-mover position and pays her bid, which is
distributed equally among the remaining players. The P\&C mechanism is then played
with the winning bidder moving first and the remaining players ordered randomly.

The bidding stage internalizes the first-mover advantage, thereby enabling ex-post
equal division of the surplus. 
We show that in any subgame perfect equilibrium of the Bid-Price-and-Choose mechanism, equilibrium payoffs coincide with the CIS value of $\tilde{v}$, and this equality holds on every equilibrium path, not merely in expectation.

\begin{prop}
Consider the Bid-Price-and-Choose mechanism with $n$ players.
\begin{enumerate}[label=\roman*)]
    \item In every subgame perfect equilibrium, players' payoffs coincide with the
    CIS value of the TU game $\tilde{v}$.
    \item The mechanism subgame-perfect implements the set of efficient alternatives
    \citep{Echenique2025}.
\end{enumerate}
\end{prop}

\begin{proof}
Part~(ii) is established in \citet{Echenique2025}. We prove part~(i).

Let $W := \arg\max_{i \in N} b_i$ denote the set of highest bidders. In any subgame
perfect equilibrium, $|W| \geq 2$: if there were a unique highest bidder, she could
profitably deviate by slightly reducing her bid while remaining the unique winner.
Hence any winner $i \in W$ must be indifferent between winning and not winning the
bidding stage:
\begin{equation*}
    \mathrm{Avg}_i + \sum_{j \in N} u_j(a^{\ast}) - \sum_{j \in N} \mathrm{Avg}_j - b
    =
    \mathrm{Avg}_i + \frac{b}{n-1}.
\end{equation*}
Solving for $b$ yields the equilibrium bid
\begin{equation*}
    b^{\ast} = \frac{n-1}{n}\!\left(\sum_{j \in N} u_j(a^{\ast}) - \sum_{j \in N} \mathrm{Avg}_j\right).
\end{equation*}
Substituting $b^{\ast}$ into the payoff expression, each player $i$ receives
\begin{equation*}
    \mathrm{Avg}_i + \frac{1}{n}\!\left(\sum_{j \in N} u_j(a^{\ast}) - \sum_{j \in N} \mathrm{Avg}_j\right),
\end{equation*}
which coincides with $\mathrm{CIS}_i(\tilde{v})$.
\end{proof}

\subsection{Price, Accept and Choose}

\subsubsection{Randomized Price, Accept and Choose}

The PA\&C mechanism gives every player the opportunity to accept or reject the previously proposed prices, thereby conferring symmetric bargaining power in a structural sense. 
In particular, any player $i \geq 2$ can reject the proposed prices and initiate a new coalition starting from herself. 
However, the gain from rejection depends on the realized order of play, which may confer an intrinsic advantage on some players relative to others. 
As in the P\&C mechanism, this asymmetry can be mitigated by randomizing the order of play prior to implementing the PA\&C mechanism.

\begin{definition}[Randomized PA\&C mechanism]
A permutation of $\{1, 2, \ldots, n\}$ is drawn uniformly at random by Nature.
The PA\&C mechanism is then played according to the realized order.
\end{definition}

We show that the expected payoffs under the randomized PA\&C mechanism coincide
with the Shapley value of the game $v$ defined in Definition~\ref{def:v}.

\begin{definition}[Shapley value: \citealp{Shapley1953}]
The Shapley value of a TU game $(N,v)$ is given, for each $i \in N$, by
\begin{equation*}
\phi_i(v)
=
\sum_{S \subseteq N \setminus \{i\}}
\frac{|S|!\,(n - |S| - 1)!}{n!}
\bigl(v(S \cup \{i\}) - v(S)\bigr).
\end{equation*}
\end{definition}

The Shapley value assigns to each player the average of her marginal contributions across all possible orderings of players.
In the PA\&C mechanism, each player effectively initiates a new coalition with all subsequent players by rejecting the preceding prices. 
By Theorem~\ref{prop:PAC_General}, the payoff of player $i$ in any SPE corresponds to her marginal contribution $x_i^{\ast}$, which captures the increment in total welfare from selecting the efficient alternative for the coalition $\{i, i+1, \dots, n\}$ relative to the coalition $\{i+1, \dots, n\}$.
From the perspective of the TU game $v$, this is precisely the marginal contribution of player $i$ to the coalition of players $j \geq i$.

By uniformly randomizing the order of players, the expected payoff of each player under the randomized PA\&C mechanism coincides with the Shapley value of the TU game $v$.

\begin{thm} \label{prop:random_PAC}
The expected payoff of any subgame perfect equilibrium of the randomized PA\&C mechanism coincides with the Shapley value of the TU game $v$.
\end{thm}

\begin{proof}
We extend the definition of the best alternative to any subset of players. For $S \subseteq N$, let
\begin{equation} \label{eq:a_S}
    \mathcal{A}_S^{\ast} = \arg\max_{a \in \mathcal{A}} \sum_{j \in S} u_j(a).
\end{equation}

By Theorem~\ref{prop:PAC_General}, in any SPE the payoff of each player $i$ coincides with $x_i^{\ast}$. Fix an order $\pi$ drawn by Nature, and let $S \subseteq N \setminus \{i\}$ be the set of players who move after $i$ (player $i$'s \emph{followers}). Applying Theorem~\ref{prop:PAC_General} (i) to the order $\pi$, player $i$'s SPE payoff is
$v(S \cup \{i\}) - v(S)$.

For any $i$ and $S \subseteq N \setminus \{i\}$, there are $|S|!$ orderings of the players in $S$ and $(n - |S| - 1)!$ orderings of the players in $N \setminus (S \cup \{i\})$, giving $|S|!\,(n - |S| - 1)!$ orders for which $i$'s follower set is exactly $S$. Each order realizes with probability $1/n!$, so player $i$'s expected payoff is
\begin{equation*}
    \sum_{S \subseteq N \setminus \{i\}} \frac{|S|!\,(n - |S| - 1)!}{n!} \bigl(v(S \cup \{i\}) - v(S)\bigr),
\end{equation*}
which coincides with $\phi_i(v)$.
\end{proof}

\subsubsection{Bid, Price, Accept and Choose}

Similarly to the P\&C case, randomizing the order of play in the PA\&C mechanism delivers an ex-ante fair division of the surplus by averaging each player's marginal contribution, but the ex-post allocation remains dependent on the realized order. 
To address this, we introduce a bidding procedure that endogenizes the order of play and ensures an order-independent ex-post allocation, in analogy with the Bid-P\&C mechanism.

The Bid-Price-Accept-and-Choose mechanism departs significantly from Bid-P\&C in its incentive structure. In the P\&C mechanism, only the first mover holds a strategic advantage and subsequent positions are inconsequential. By contrast, every position in the PA\&C mechanism affects a player's capacity to extract surplus. 
Our bidding procedure addresses this positional asymmetry by compensating players for the first-mover position, following the approach of \citet{PEREZCASTRILLO2001}.

We first describe the two-player case and then extend to $n$ players.

\paragraph{Two-Player Case}

The bidding procedure compensates players for being selected as the first mover. By Proposition~\ref{prop:SPE_PAC_2}, in any SPE of the PA\&C mechanism each player receives her marginal contribution. Hence, if player $i$ is the first mover, she obtains $u_1(a^{\ast}) + u_2(a^{\ast}) - u_j(a_j^{\ast})$, since she must guarantee player $j$ at least her standalone payoff. Here, $a^{\ast}$ denotes the efficient alternative selected in equilibrium, and $a_j^{\ast}$ (resp. $a_i^{\ast}$) denotes the standalone efficient alternative for player $j$ (resp. player $i$). If instead player $i$ is the second mover, her payoff is $u_i(a_i^{\ast})$. Being the second mover is thus advantageous, as
\begin{equation*}
u_1(a^{\ast}) + u_2(a^{\ast}) \leq u_1(a_1^{\ast}) + u_2(a_2^{\ast}).
\end{equation*}
Then, by Definition~\ref{def:v}, the game $v$ is subadditive: $v(N) \le v(\{1\}) + v(\{2\})$.
We leverage this property to design a bidding game that compensates players for not being the second mover. The procedure is as follows:

\begin{itemize}
\item[] \textbf{First Stage (bid for the first-mover position):} Each player $i$ simultaneously submits a bid $b^j_i$, representing the compensation she demands from player $j$ for being selected as the first mover. The player with the lowest net bid $B_i$, defined as $B_i = b^j_i - b^i_j$, becomes the first mover. In case of a tie, the winner is selected randomly. The first mover receives the requested compensation $b^j_i$ paid by player $j$.
\item[] \textbf{Second Stage:} The PA\&C mechanism is then played in the order determined in the first stage.
\end{itemize}

The net bid $B_i$ captures each player's willingness to avoid being the first mover: it reflects how much more she demands to move first relative to her opponent.\footnote{This interpretation of the net bid as compensation follows \citet{PEREZCASTRILLO2001}, who study weakly superadditive games.} Unlike \citet{PEREZCASTRILLO2001}, who analyze the division of a fixed, exogenously given surplus, our setting is novel in that players bid to determine their positions in a dynamic mechanism that endogenously generates the surplus through their strategic choices. Designing such a bidding procedure for PA\&C is not straightforward: unlike in P\&C, where only the first mover's position matters, every position in PA\&C influences payoffs. Our construction ensures that positional rents are properly compensated while preserving efficiency.

We now show that in the Bid-PA\&C mechanism with two players, players' payoffs coincide with the standard solution, hence with the Shapley value of the TU game $v$.

\begin{prop} \label{prop:Bid_2}
Consider the Bid-PA\&C mechanism with two players. Then in any SPE, the payoffs of the players coincide with the Shapley value of the TU game $v$.
\end{prop}

\begin{proof}
We proceed in three steps: we construct a symmetric bid $b$ that implements the Shapley value (Step 1), show that $b$ constitutes an equilibrium bid (Step 2), and establish that it is the unique equilibrium bid (Step 3).

\textbf{Step 1: Constructing a symmetric bid.} 
Let $b_i^j$ denote the compensation bid that player $i$ requests from player $j (\neq i)$ in exchange for being selected as the first mover.
Suppose each player $i$ submits
\begin{equation*}
b_i^j = u_j(a^{\ast}_j) - \frac{1}{2}\big(u_1(a^{\ast}) + u_2(a^{\ast}) - u_i(a^{\ast}_i) + u_j(a^{\ast}_j)\big).
\end{equation*}
Then $b_i^j = b_j^i = b$, where
\begin{equation*}
b = \frac{1}{2}\big(u_1(a^{\ast}_1) + u_2(a^{\ast}_2) - u_1(a^{\ast}) - u_2(a^{\ast})\big),
\end{equation*}
so the net bid is zero for both players. By Proposition~\ref{prop:SPE_PAC_2}, in any SPE of the PA\&C mechanism each player's payoff coincides with her marginal contribution. If player $i$ is the first mover, her payoff is
\begin{equation*}
u_1(a^{\ast}) + u_2(a^{\ast}) - u_j(a^{\ast}_j) + b = \frac{1}{2}\big(u_1(a^{\ast}) + u_2(a^{\ast}) - u_j(a_j^{\ast}) + u_i(a_i^{\ast})\big) = \phi_i(v),
\end{equation*}
and if instead she is the second mover, her payoff is
\begin{equation*}
u_i(a^{\ast}_i) - b = \frac{1}{2}\big(u_1(a^{\ast}) + u_2(a^{\ast}) - u_j(a_j^{\ast}) + u_i(a_i^{\ast})\big) = \phi_i(v).
\end{equation*}
Both roles thus yield the same payoff $\phi_i(v)$.

\textbf{Step 2: $b$ is an equilibrium bid.} If player $i$ raises her bid above $b$, she becomes the second mover and still receives $\phi_i(v)$. If she lowers her bid below $b$, she becomes the first mover but receives less than $\phi_i(v)$. Hence no player has a profitable deviation, and $b$ is an equilibrium bid.

\textbf{Step 3: Uniqueness.} In equilibrium, each player must be indifferent between being the first mover or not, and both players must submit equal bids to keep the net bid at zero; otherwise some player would have a profitable deviation. Hence $b_i^j = b_j^i = b'$ in equilibrium, and indifference requires
\begin{equation*}
u_1(a^{\ast}) + u_2(a^{\ast}) - u_j(a^{\ast}_j) + b' = u_i(a^{\ast}_i) - b'.
\end{equation*}
Solving yields $b' = b$, so $b$ is the unique equilibrium bid.
\end{proof}

\paragraph{Many-Player Case}
For $n$ players, we extend the bidding procedure sequentially. The first stage determines the order of play through a sequence of bidding rounds analogous to the two-player case; the second stage then implements the PA\&C mechanism according to this order.
\begin{itemize}
\item[] \textbf{First Stage:} The first stage consists of $n-1$ rounds that determine the order in which players move.
\begin{itemize}
\item[i)] \textit{First mover.} Each player $i$ simultaneously submits a bid vector $b_i = (b_i^j)_{j \neq i}$, where $b_i^j$ represents the compensation player $i$ requests if selected as the first mover, to be paid by player $j$. The player with the lowest net bid $B_i = \sum_{j \neq i} b_i^j - \sum_{j \neq i} b_j^i$ becomes the first mover. In case of a tie, the winner is selected randomly. The first mover receives total compensation $\sum_{j \neq i} b_i^j$, paid by the remaining players.
\item[ii)] \textit{Subsequent positions.} The remaining players repeat this procedure to determine the second mover, and so on, until all positions are assigned.
\end{itemize}
\item[] \textbf{Second Stage:} The PA\&C mechanism is then played in the order determined in the first stage.
\end{itemize}

The intuition mirrors the two-player case. Players bid sequentially from the first to the last position; the player ultimately selected as the last mover pays compensation to the others but is then entitled to $u_i(a_i^{\ast})$ when playing the subsequent PA\&C mechanism. The earlier rounds of bidding compensate players for moving earlier in the order, since the game $v$ is subadditive.\footnote{For any disjoint $S,T \subseteq N$, we have 
$\max_{a \in \mathcal{A}} \sum_{i \in S \cup T} u_i(a) = \max_{a \in \mathcal{A}} \left( \sum_{i \in S} u_i(a) + \sum_{i \in T} u_i(a) \right) \le \max_{a \in \mathcal{A}} \sum_{i \in S} u_i(a) + \max_{a \in \mathcal{A}} \sum_{i \in T} u_i(a)$.
} 
By leveraging this bidding design, we show that in any SPE, the resulting payoffs coincide with the Shapley value of $v$.

\begin{thm}
Consider the Bid-PA\&C mechanism with $n$ players. In any subgame perfect equilibrium, the payoffs coincide with the Shapley value of the TU game $v$.
\end{thm}
\begin{proof}
We proceed by induction on $n$. The base case $n = 2$ follows from Proposition~\ref{prop:Bid_2}. Assume the claim holds for all games with at most $n-1$ players, where $n \geq 3$. We prove it for $n$ players.

Before presenting the argument, we recall the following recursive formulation of the Shapley value:
\begin{equation} \label{eq:recursive_Sh}
\phi_i(N,v) = \frac{1}{n}\big(v(N) - v(N \setminus \{i\})\big) + \frac{1}{n}\sum_{j \neq i} \phi_i(N \setminus \{j\}, v).
\end{equation}
where, with a slight abuse of notation, $(N \setminus \{j\}, v)$ denotes the game with player set $N \setminus \{j\}$ and characteristic function $v$ restricted to subsets of $N \setminus \{j\}$.
The first term on the right-hand side represents player $i$'s marginal contribution when she is the last to join $N$, while the second term averages over the cases where some other player $j$ is the last to join. Each case occurs with equal probability $1/n$, reflecting the uniform distribution over orderings.

We establish the claim in four steps. We construct a bid vector for each player $i$ that implements the Shapley value of $v$ and satisfies $B_i = 0$ (Step 1), show that this bid constitutes an equilibrium of the bidding game (Step 2), prove that in any equilibrium the net bid of every player is zero (Step 3), and conclude that every equilibrium yields the Shapley value (Step 4).

\textbf{Step 1: Constructing bid vectors that implement the Shapley value.} For each $i$ and $j \neq i$, let
\begin{equation*}
b_i^j = \phi_j\big(N \setminus \{i\}, v\big) - \phi_j(N, v).
\end{equation*}
The net bid of player $i$ is
\begin{equation*}
B_i = \sum_{j \neq i}\big(\phi_j(N \setminus \{i\}, v) - \phi_j(N, v)\big) - \sum_{j \neq i}\big(\phi_i(N \setminus \{j\}, v) - \phi_i(N, v)\big).
\end{equation*}
By the balanced contributions property of the Shapley value \citep{Myerson_IJGT},
\begin{equation*}
\phi_j(N \setminus \{i\}, v) - \phi_j(N, v) = \phi_i(N \setminus \{j\}, v) - \phi_i(N, v) \qquad \text{for all } i, j,
\end{equation*}
which implies $B_i = 0$ for all $i$.

Suppose these bids constitute an equilibrium and let $k \in N$ be the randomly chosen first mover. By Theorem~\ref{prop:PAC_General}, each player's SPE payoff in the PA\&C mechanism coincides with her marginal contribution. Since the bids compensate players exactly for positional rents, each player obtains her Shapley value regardless of the realized order. Specifically, if player $i$ is not the first mover, her payoff is
\begin{equation*}
\phi_i\big(N \setminus \{k\}, v\big) - \big(\phi_i\big(N \setminus \{k\}, v\big) - \phi_i(N, v)\big) = \phi_i(N, v).
\end{equation*}
If instead player $i$ is the first mover, her payoff is
\begin{equation*}
v(N) - v\big(N \setminus \{i\}\big) + \sum_{j \neq i}\big(\phi_i\big(N \setminus \{j\}, v\big) - \phi_i(N, v)\big) = n\phi_i(N, v) - (n-1)\phi_i(N, v) = \phi_i(N, v),
\end{equation*}
where the equality follows from the recursive formula~\eqref{eq:recursive_Sh}. Thus the proposed bids implement the Shapley value for every player.

\textbf{Step 2: The bid vectors constitute an equilibrium.} By Step 1, with the above bids each player earns her Shapley value regardless of whether she is selected as first mover. Consider a unilateral deviation by player $i$. If player $i$ raises her bids, she is no longer selected as first mover, and her payoff remains equal to her Shapley value. If instead she lowers her bids, she is selected as first mover but receives less compensation, yielding a payoff strictly below her Shapley value. Since no deviation yields a strictly higher payoff, the proposed bid vectors form an equilibrium.

\textbf{Step 3: In any equilibrium, $B_i = 0$ for all $i \in N$.} Let $W = \arg\min_{j \in N} B_j$ be the set of players with the lowest net bid. We show that $W = N$ in equilibrium. Suppose not, so $|W| < |N|$. Take any $i \in W$ and $k \notin W$, and consider the following deviation by $i$: she lowers her bid to each $j \in W \setminus \{i\}$ by $\varepsilon > 0$ and raises her bid to player $k$ by $|W|\varepsilon$. Formally:
\begin{equation*}
\begin{cases}
b_i^{\prime j} = b_i^j - \varepsilon & \text{for all } j \in W \setminus \{i\}, \\
b_i^{\prime k} = b_i^k + |W|\varepsilon & \text{for } k \notin W, \\
b_i^{\prime l} = b_i^l & \text{for all } l \in N \setminus (W \cup \{k\}).
\end{cases}
\end{equation*}
For $\varepsilon$ small enough, the set of winners remains unchanged: the net bids of all $j \in W \setminus \{i\}$ increase by $\varepsilon$, as does the net bid of $i$, while the net bid of $k$ decreases by $|W|\varepsilon$. Since $i$ remains among the winners but earns strictly more if selected, while her payoff if not selected is unchanged, this is a profitable deviation, a contradiction. Hence $W = N$.

Since $W = N$, all players share the same net bid: $B_i = B_j$ for all $i, j \in N$. Moreover,
\begin{equation*}
\sum_{i \in N} B_i = \sum_{i \in N}\left(\sum_{j \neq i} b_i^j - \sum_{j \neq i} b_j^i\right) = 0,
\end{equation*}
since each $b_i^j$ appears once positively and once negatively. As all $B_i$ are equal and sum to zero, $B_i = 0$ for all $i \in N$.

\textbf{Step 4: In any equilibrium, each player receives her Shapley value.} Since all $B_i = 0$ by Step 3, each player must be indifferent among all possible winners. If player $i$ strictly preferred to win, she could lower her bids slightly to become the unique winner and improve her payoff; if she strictly preferred some player $j$ to win, she could slightly increase her bid $b_i^j$ to that effect. Neither occurs in equilibrium, so every player is indifferent among winners.

Let $u_i^j$ denote the payoff of player $i$ when player $j$ is selected as the first mover. When player $i$ is the first mover, her payoff is
\[
u_i^i = v(N) - v(N \setminus \{i\}) + \sum_{j \neq i} b_i^j.
\]
When player $j \neq i$ is the first mover, player $i$'s payoff is
\[
u_i^j = \phi_i(N \setminus \{j\}, v) - b_j^i.
\]
Summing over all possible winners:
\begin{equation*}
v(N) - v(N \setminus \{i\}) + \sum_{j \neq i} b_i^j + \sum_{j \neq i}\big(\phi_i(N \setminus \{j\}, v) - b_j^i\big) = n\phi_i(N, v) + B_i = n\phi_i(N, v),
\end{equation*}
where the last equality uses $B_i = 0$. Since player $i$ is indifferent among all winners, $u_i^j = u_i^k$ for all $j, k \in N$, and therefore $u_i^j = \phi_i(N, v)$ for every $j$. Hence every player receives her Shapley value in equilibrium.
\end{proof}

\section{Other Mechanisms and TU-Game Solutions}\label{sec:Further}

We now examine an additional design to show that the mapping between mechanisms and TU-game solutions extends beyond P\&C and PA\&C. 
We also consider the surplus distribution achieved via the VCG mechanism. 
{Unlike P\&C and PA\&C, the VCG mechanism is a static, dominant-strategy incentive-compatible mechanism rather than a sequential bargaining protocol, so it provides a benchmark independent of the sequential logic underlying our main results. 
Since VCG payments are not budget balanced in general, we study their balanced version and show that it too corresponds to a specific, well-known solution concept, the Egalitarian Non-Separable Contribution (ENSC), the dual of the CIS value \citep{DriessenFunaki1991}. 
This indicates that the correspondence between efficient mechanisms and TU-game solution concepts is not confined to sequential pricing protocols, but extends to a canonical mechanism-design benchmark as well.}

\subsection{Price, Accept and Choose Only (PA\&CO)}

One of the major differences between PA\&C and P\&C is that players can accept or reject the previous price vector. 
Since the last player also has this right, she need not be indifferent among all alternatives in equilibrium. 
One might wonder what happens if this last stage is altered so that the last player cannot accept or reject, but can only choose an alternative. 
This seemingly minor change alters the core argument in the selection phase. 
When the last player cannot accept or reject the previous price vectors, the optimal strategy for the preceding player is to make the last mover indifferent among all alternatives; this forces the last mover to receive her average utility in equilibrium, since otherwise the preceding player could slightly adjust prices to improve her own payoff. 
By contrast, when the last mover can accept or reject the previous price vector (as in PA\&C), she is not indifferent across all alternatives in equilibrium, but only between accepting (and selecting the efficient alternative) and rejecting (and selecting her own favorite alternative). In the following proposition, we show that this modification changes the associated TU-game, which becomes $\tilde{v}$, but not the solution concept, which remains the Shapley value.

Specifically, the PA\&CO mechanism proceeds as follows:
 
\begin{itemize}
 \item [i)] Player 1 proposes a price vector $p_1 \in P$.
\item [ii)] For each $i = 2, \dots, n-1$, player $i$ either accepts or rejects all previously proposed price vectors $\left(p^c_{j}\right)_{j<i}$, and subsequently proposes to player $i+1$ a price vector $p_i \in P$.
\item [iii)] Player $n$ chooses an alternative $a \in \mathcal{A}$.
\item [iv)] Each player $i>1$ transfers $p^c_{i-1}(a)$ to player $i-1$.
\end{itemize}

In the PA\&CO mechanism, when player $i$ rejects, it is as if the PA\&CO mechanism restarts at player $i$. If player $i$ accepts, the PA\&CO mechanism instead continues from the earliest preceding player whose price vector has not been rejected.

\bigskip

In order to associate the SPE payoffs of the PA\&CO mechanism with a solution concept of a TU game, we use the game $\tilde{v}$ already introduced in Definition~\ref{def:v_tilde}, with $\tilde{v}(\{i\}) = \mathrm{Avg}_i$ for each $i \in N$ and $\tilde{v}(S) = \max_{a \in \mathcal{A}} \sum_{i \in S} u_i(a)$ for all $S \subseteq N$ with $|S| \geq 2$.

Let
 \begin{equation*}
     \tilde{x}_i^{\ast}=
     \begin{cases}
        \tilde{v}\left(\{i,\dots,n\}\right)- \tilde{v}\left(\{i+1,\dots,n\}\right) & \text{ for } i=1,\dots,n-1, \\
        \tilde{v}\left(\{n\}\right) & \text{ for } i=n.
     \end{cases}
 \end{equation*}

\begin{prop} \label{prop:PACO_General}
Consider the PA\&CO mechanism with $n \ge 2$ players. Then:
\begin{itemize}
    \item [i)] In any SPE, the equilibrium payoff of each player $i$ coincides with her marginal contribution $\tilde{x}_{i}^{\ast}$, with respect to the game $\tilde{v}$.
    \item [ii)] The PA\&CO mechanism subgame-perfect implements the set of efficient alternatives.
\end{itemize}
\end{prop}

\begin{proof}
The proof proceeds by induction on $n$, mirroring the proof of Proposition~\ref{prop:PAC_General}. The base case $n=2$ coincides with the P\&C mechanism, since player $n$ has no accept/reject option and there is no intermediate player; Proposition~\ref{prop:PC_General} then gives $y_1 = v(N) - \mathrm{Avg}_2 = \tilde{x}_1^*$ and $y_2 = \mathrm{Avg}_2 = \tilde{x}_2^*$. Now, assume the proposition holds for all games with fewer than $n$ players, $n \geq 3$.

Fix an arbitrary SPE of the $n$-player PA\&CO mechanism, with equilibrium payoff vector $y = (y_1, \ldots, y_n)$ and equilibrium outcome $a$. Since transfers cancel along the chain,
\begin{equation}
\sum_{i \in N} y_i = \sum_{i \in N} u_i(a). \label{eq:PACO_sum_payoffs}
\end{equation}

\textbf{Claim 1:} $y_i \geq \tilde{x}_i^{\ast}$ for each $i = 2, \ldots, n-1$.

If player $i$ rejects, the continuation game is a PA\&CO mechanism among players $i, \ldots, n$. By the induction hypothesis, player $i$ obtains $\tilde{x}_i^{\ast}$ in that subgame, which depends only on $\tilde{v}$ restricted to $\{i, \ldots, n\}$, and so coincides with $\tilde{x}_i^{\ast}$ in the original game. So, she can guarantee at least $\tilde{x}_i^{\ast}$.

\textbf{Claim 2:} $y_n \geq \tilde{x}_n^{\ast}$.

Player $n$ has no accept/reject option: she always faces the price vector $p_{n-1} \in P$ proposed by player $n-1$, whatever it is. Since $p_{n-1}$ is budget balanced,
\[
y_n = \max_{a \in \mathcal{A}} \bigl(u_n(a) - p_{n-1}(a)\bigr)
\geq \frac{1}{m} \sum_{a \in \mathcal{A}} \bigl(u_n(a) - p_{n-1}(a)\bigr)
= \mathrm{Avg}_n - \frac{1}{m} \sum_{a \in \mathcal{A}} p_{n-1}(a)
= \mathrm{Avg}_n = \tilde{x}_n^{\ast}.
\]

\textbf{Claim 3:} $y_1 \geq \tilde{x}_1^{\ast}$.

Suppose for contradiction that $y_1 < \tilde{x}_1^{\ast}$, and let $\varepsilon := \tilde{x}_1^{\ast} - y_1 > 0$. Fix any $a^{\ast} \in \mathcal{A}^{\ast}$ and define
\[
u_E := \sum_{j \geq 2} u_j(a^{\ast}) - \sum_{j \geq 2} \tilde{x}_j^{\ast}.
\]
Since $\sum_{j \geq 2} \tilde{x}_j^{\ast} = \tilde{v}(\{2, \ldots, n\}) = \max_{a \in \mathcal{A}} \sum_{j \geq 2} u_j(a) \geq \sum_{j \geq 2} u_j(a^{\ast})$, we have $u_E \leq 0$. Consider the deviation $\tilde{p}_1 \in P$ defined by
\[
\tilde{p}_1(a^{\ast}) = u_E - \frac{\varepsilon}{2}, \
\tilde{p}_1(a) = -\frac{1}{m-1}\Bigl(u_E - \frac{\varepsilon}{2}\Bigr) \ \text{ for } a \neq a^{\ast}.
\]
Since $u_E \leq 0$, we have $\tilde{p}_1(a) > 0$ for all $a \neq a^{\ast}$, making $a^{\ast}$ the unique maximizer of $\sum_{j \geq 2} u_j(a) - \tilde{p}_1(a)$. By the induction hypothesis applied to the continuation game after player 2 accepts $\tilde{p}_1$, player 2 accepts (since doing so yields strictly more than $\tilde{x}_2^{\ast}$) and $a^{\ast}$ is selected. A direct calculation gives that player 1's payoff from $\tilde{p}_1$ is $\tilde{x}_1^{\ast} - \frac{\varepsilon}{2} > y_1$, contradicting equilibrium. Hence $y_1 \geq \tilde{x}_1^{\ast}$.

\textbf{Claim 4:} $y_i = \tilde{x}_i^{\ast}$ for all $i \in N$.

By~\eqref{eq:PACO_sum_payoffs} and efficiency of $\mathcal{A}^{\ast}$,
\[
\sum_{i \in N} y_i = \sum_{i \in N} u_i(a) \leq \sum_{i \in N} u_i(a^{\ast}) = \sum_{i \in N} \tilde{x}_i^{\ast}.
\]
Combined with $y_i \geq \tilde{x}_i^{\ast}$ for all $i$ (Claims 1--3), this forces $y_i = \tilde{x}_i^{\ast}$ for every $i \in N$, proving part i).

To prove part ii), we show that every SPE selects some $a^{\ast} \in \mathcal{A}^{\ast}$, and that every $a^{\ast} \in \mathcal{A}^{\ast}$ is selected by some SPE.

The first part follows immediately: if the equilibrium outcome $a \notin \mathcal{A}^{\ast}$, then $\sum_{i \in N} u_i(a) < \sum_{i \in N} \tilde{x}_i^{\ast}$, contradicting Claim 4 via~\eqref{eq:PACO_sum_payoffs}.

For the second part, fix any $a^{\ast} \in \mathcal{A}^{\ast}$ and consider strategies under which ties among efficient alternatives are broken in favor of $a^{\ast}$. Since all alternatives in $\mathcal{A}^{\ast}$ generate the same total surplus, this tie-breaking leaves every player's marginal contribution unchanged. By Claim 4, no player has a profitable deviation, so the resulting strategy profile is an SPE that selects $a^{\ast}$.
\end{proof}

\begin{prop} \label{prop:PACO_Shapley}
The expected payoff of any SPE of the randomized PA\&CO mechanism coincides with the Shapley value of the TU game $\tilde{v}$.
\end{prop}
\begin{proof}
    The proof mirrors exactly the one in Theorem \ref{prop:random_PAC}.
\end{proof}

\subsection{Balanced VCG Payments}
While the PA\&C mechanism, together with its extensions, and the P\&C mechanism are designed to subgame-perfect implement efficient alternatives, it is crucial to understand how the surplus is distributed among players compared to other standard efficient mechanisms. A widely recognized tool for ensuring {efficiency} is the VCG mechanism. When considering the outcomes generated by the efficient alternative together with the balanced VCG payments, we show that the resulting payoffs align with the ENSC solution of the TU games $v$ and $\tilde{v}$. This contrasts with the (expected) equilibrium payoffs of the (randomized) Bid-PA\&C mechanism, which we have shown to coincide with the Shapley value of the TU game $v$. Similarly, the (expected) payoffs of the (randomized) Bid-P\&C mechanism correspond to the CIS value of the TU game $\tilde{v}$. 

More formally, let $t_{i}^{V}$ be the VCG payment of $i$ defined by:
\[
t_{i}^{V}=\sum_{j\neq i}u_{j}\left( a^{\ast }\left( N\backslash \left\{
i\right\} \right) \right) -\sum_{j\neq i}u_{j}\left( a^{\ast }\left(
N\right) \right) .
\]%
Using the language of $v$, we can write it as:%
\begin{equation}
t_{i}^{V}=v\left( N\backslash \left\{ i\right\} \right) -v\left( N\right)
+u_{i}\left( a^{\ast }\left( N\right) \right) .  \label{eq:t_i^V}
\end{equation}%
As is well known, the VCG payment is not balanced in general.

Let $t_{i}^{b}$\ denote the balanced VCG payment of $i$. Then, using a
constant $c$, $t_{i}^{b}$ is defined by the following conditions:%
\[
\left\{ 
\begin{array}{l}
t_{i}^{b}=t_{i}^{V}+c\text{ for all }i, \\ 
\sum_{j\in N}t_{j}^{b}=0.%
\end{array}%
\right. 
\]%
Therefore, $c$ is determined by:%
\[
c=-\frac{1}{n}\sum_{j\in N}t_{j}^{V}.
\]%
By (\ref{eq:t_i^V}), 
\begin{eqnarray}
c &=&-\frac{1}{n}\sum_{j\in N}\left\{ v\left( N\backslash \left\{ j\right\}
\right) -v\left( N\right) +u_{j}\left( a^{\ast }\left( N\right) \right)
\right\}   \nonumber \\
&=&-\frac{1}{n}\sum_{j\in N}\left\{ v\left( N\backslash \left\{ j\right\}
\right) +u_{j}\left( a^{\ast }\left( N\right) \right) \right\} +v\left(
N\right)   \nonumber \\
&=&-\frac{1}{n}\sum_{j\in N}v\left( N\backslash \left\{ j\right\} \right) +%
\frac{n-1}{n}v\left( N\right) .  \label{eq:c_by_v}
\end{eqnarray}

\begin{prop} \label{prop:balanced_VCG_ENSC}
Allocation by the efficient alternative and the balanced VCG payments coincides with the Egalitarian Non-Separable Contribution values of the games $v$ and $\tilde{v}$.
\end{prop}

\begin{proof}
ENSC is defined by equally sharing the remainder of the grand coalition value
after distributing the separable contribution of each player, that is,%
\begin{equation}
ENSC_{i}\left( v\right) =SC_{i}\left( v\right) +\frac{1}{n}\left( v\left(
N\right) -\sum_{j\in N}SC_{j}\left( v\right) \right)   \label{eq:ENSCdef}
\end{equation}%
where $SC_{i}\left( v\right) $ stands for separable contribution of $i$,
defined as:%
\[
SC_{i}\left( v\right) =v\left( N\right) -v\left( N\backslash \left\{
i\right\} \right) .
\]%
By summing up over all individuals, 
\[
\sum_{j\in N}SC_{j}\left( v\right) =nv\left( N\right) -\sum_{j\in N}v\left(
N\backslash \left\{ j\right\} \right) .
\]%
Plugging this into (\ref{eq:ENSCdef}), we obtain:%
\begin{eqnarray*}
ENSC_{i}\left( v\right)  &=&SC_{i}\left( v\right) +\frac{1}{n}\left( \left(
1-n\right) v\left( N\right) +\sum_{j\in N}v\left( N\backslash \left\{
j\right\} \right) \right)  \\
&=&SC_{i}\left( v\right) -\frac{n-1}{n}v\left( N\right) +\frac{1}{n}%
\sum_{j\in N}v\left( N\backslash \left\{ j\right\} \right) .
\end{eqnarray*}%
By (\ref{eq:c_by_v}), we have:%
\[
ENSC_{i}\left( v\right) =SC_{i}\left( v\right) -c.
\]

Now, let us consider the payoff of individual $i$ when the allocation by the
efficient alternative and balanced VCG payments is implemented:
\begin{eqnarray*}
u_{i}\left( a^{\ast }\left( N\right) \right) -t_{i}^{b} &=&u_{i}\left(
a^{\ast }\left( N\right) \right) -\left( t_{i}^{V}+c\right)  \\
&=&-v\left( N\backslash \left\{ i\right\} \right) +v\left( N\right) -c \\
&=&SC_{i}\left( v\right) -c.
\end{eqnarray*}%
The first and the third equalities are by definition. The second one is by (\ref{eq:t_i^V}). 
The last line coincides with the ENSC. Since the values of $v$ and $\tilde{v}$ differ only for the standalone coalitions, their ENSC values coincide as well.
\end{proof}

\section{Conclusion}\label{sec:Conclusion}

This paper shows that simple sequential pricing mechanisms can achieve efficiency and generate surplus divisions that correspond to canonical TU-game solution concepts. The Price-and-Choose mechanism yields the Center of the Imputation Set value; our proposed extension, Price-Accept-and-Choose, yields the Shapley value. 
The difference traces to a simple change in the design:  whether players can reject proposed prices, 
which changes each player's standalone worth in the corresponding TU game, from her average utility across alternatives to her maximum one, and thereby changes the equilibrium division of the entire surplus.
Bidding procedures that endogenize the order of play implement these TU-game solutions on every equilibrium path, not merely in expectation.

These results suggest a broader agenda. From a design perspective, our findings imply that an arbitrator choosing among efficient mechanisms is implicitly choosing a TU-game solution concept. This offers a normative criterion for mechanism selection: if the Shapley value's marginal-contribution logic is deemed fair, PA\&C is the appropriate institution; if equal sharing of the surplus is preferred, P\&C delivers the CIS value. Making this mapping explicit gives practitioners a principled basis for choosing among mechanisms that are otherwise observationally equivalent in their efficiency properties.

Several directions remain open. First, the difference between P\&C and PA\&C hinges on a single design choice: whether players can reject proposed prices. In the current PA\&C mechanism, rejection means that all previously proposed prices are deleted. A natural extension would allow players to opt out of the transfer system individually while the remaining players continue to participate under the prices they have accepted. This selective exit option would be particularly relevant in partnership dissolution problems, where some partners may wish to withdraw from a collectively negotiated restructuring while others proceed (see \cite{Cramton1987} for a classic paper and \cite{VanEssen2016} for a recent paper on partnership dissolution). 
The resulting mechanism would sit between P\&C (no rejection) and PA\&C (full rejection), and the corresponding TU-game solution would presumably lie between the CIS and the Shapley values. Compromises of this kind have been studied in cooperative game theory: e.g., the consensus value of \citet{JuBormRuys2007}. 
Identifying a sequential pricing mechanism whose equilibrium payoffs implement such a compromise would turn the pair of benchmarks described in this paper into a full spectrum, with the extent of the players' opt-out rights as the parameter that moves along it.

Second, our mechanisms require quasi-linear preferences and unrestricted transfers. Budget constraints or wealth effects would break the clean marginal-contribution characterization and could reintroduce inefficiency. 

Third, while \cite{Funaki_Exp} provide initial experimental evidence, further laboratory and field tests could assess whether the theoretical surplus divisions emerge in practice and whether subjects find PA\&C's rejection option as strategically transparent as the theory assumes.

\subsection*{Statements and Declarations}

\subsubsection*{Funding}
This work has benefited from a State grant managed by the Agence Nationale de la Recherche under the Investissements d'Avenir programme with the reference ANR-18-EURE-0005/EUR DATA EFM and ANR-11-IDEX-0003/Labex Ecodec/ANR-11-LABX-0047.

\subsubsection*{Competing Interests}
The authors have no competing financial or non-financial interests to declare that are relevant to the content of this article. Giacomo Rostagno is employed by RBB Economics; the views expressed are those of the author and do not reflect those of RBB Economics or its clients.

\subsubsection*{Author Contributions}
All authors contributed to the study conception and design, theoretical analysis, and writing of the manuscript. All authors read and approved the final manuscript.

\subsubsection*{Data Availability}
Data sharing is not applicable to this article as no datasets were generated or analyzed during the current study.

\bibliography{bibbib}
\end{document}